\documentclass[journal,10pt,twocolumn]{IEEEtran}
\usepackage{geometry}
\usepackage{diagbox}
\usepackage{multirow}
\usepackage{booktabs}
\usepackage{graphicx}
\usepackage[fleqn]{amsmath}
\usepackage{cases}
\usepackage{subfigure}
\usepackage{stfloats}
\usepackage{epstopdf}
\usepackage{cite}
\usepackage{algorithm}
\usepackage{algorithmic}
\usepackage{tabu}

\usepackage{amssymb}
\usepackage{color}
\usepackage{xcolor}

\newtheorem{definition}{Definition}

\newtheorem{remark}{Remark}
\newtheorem{lemma}{Lemma}
\newtheorem{corollary}{Corollary}
\newtheorem{theorem}{\textbf{Theorem}}
\newtheorem{assumption}{Assumption}

\begin{document}
\title{Covariance-Intersection-based Distributed Kalman Filtering: Stability Problems Revisited}

\author{Zhongyao~Hu, Bo Chen, Chao Sun, Li Yu
\thanks{Manuscript received xx; accepted xx. This paper was recommended by Associate Editor xx. {\em (Corresponding author: Bo Chen.)}}
\thanks{Z. Hu, B. Chen, J. Wang, C. Sun and L. Yu are with Department of Automation, Zhejiang University of Technology, Hangzhou 310023, China. (email: bchen@aliyun.com).}
}
\markboth{Preprint}
{Shell \MakeLowercase{\textit{et al.}}: Bare Demo of IEEEtran.cls for Journals}
\maketitle

\begin{abstract}
This paper studies the stability of covariance-intersection (CI)-based distributed Kalman filtering in time-varying systems. For the general time-varying case, a relationship between the error covariance and the observability Gramian is established. Utilizing this relationship, we demonstrate an intuition that the stability of a node is only related to the observability of those nodes that can reach it uniformly. For the periodic time-varying case, it is proved by a monotonicity analysis method that CI-based distributed Kalman filtering converges periodically for any initial condition. The convergent point is shown to be the unique positive definite solution to a Riccati-like equation. Additionally, by constructing an intermediate difference equation, the closed-loop transition matrix of the estimation error system is proved to be Schur stable. Notably, all theoretical results are obtained without requiring network connectivity assumptions. Finally, simulations verify the effectiveness of the stability results.
\end{abstract}

\begin{IEEEkeywords}
Distributed estimation, Kalman filtering, stability analysis, covariance intersection (CI).
\end{IEEEkeywords}

\vspace{-3pt}
\section{Introduction}
To effectively control and make decisions regarding the objective system, it is essential to estimate the system state. Centralized estimation methods, such as Kalman filtering, have received a lot of attention in the past, where they integrate all sensor data into a fusion center for processing \cite{Optimalfiltering}. Nevertheless, cyber-physical systems with large spatial spans, such as power grids, are often monitored by networks with a large number of local sensors \cite{HUANG20221}. In such case, the centralized manner is no longer applicable. Differently, distributed manners allocate calculation and communication tasks to the whole sensor network, thereby exhibiting better robustness and reliability \cite{REGO201936}. Consensus-based methods are one of the earliest distributed estimation methods \cite{Olfati5399678,MATEI20121776,Battistelli6897960,ZHANG201934}. These methods build a local filter in each node, and then guarantees that all nodes reach a consensus on the system state through sufficient number of consensus iterations. Based on a distributed observability decomposition technique, \cite{Mitra8270595} and \cite{Mitra9627574} proposed a distributed observer by explicitly characterizing the state component contained in each node. Moreover, \cite{Yan9778959} proposed a distributed estimator by reformulating the centralized steady-state Kalman filtering as a linear combination of multiple subsystems. However, this method requires that all measurements be scalar. Particularly, an inherent challenge in the distributed estimation filed is the difficulty in accurately obtaining correlations among nodes. To address this issue, \cite{BATTISTELLI2014707} employed the covariance intersection (CI) fusion rule \cite{Julier609105}, leading to the development of CI-based distributed Kalman filtering. Specifically, CI-based distributed Kalman filtering constructs a Kalman filter at each node and fuses information from neighboring nodes using the CI fusion rule. Due to the utilization of the CI fusion rule, this method not only ensures the consistency of the estimate, but also is simple and efficient. 

Stability is a fundamental property that every distributed estimator must guarantee. Considerable research has been devoted to the stability of CI-based distributed Kalman filtering \cite{BATTISTELLI2014707,BATTISTELLI2016169,Liu8113573,CHEN2021109769,Li8845692,Wei8423433,Qian9599480,Qian10423131}. In \cite{BATTISTELLI2014707}, CI-based distributed Kalman filtering was proved to be stable if the system is collectively observable and the communication graph is strongly connected. By defining the observability of nonlinear systems via Jacobi matrices, \cite{BATTISTELLI2016169} generalized the stability result in \cite{BATTISTELLI2014707} to the nonlinear case. In \cite{Liu8113573} and \cite{CHEN2021109769}, the stability of CI-based distributed Kalman filtering in the presence of packet loss was studied. \cite{Li8845692} and \cite{Wei8423433} proved that CI-based distributed Kalman filtering is stable when the system is collectively detectable. Additionally, \cite{Qian9599480} showed that even if the uniform observability only holds when the system dynamics decays slowly, CI-based distributed Kalman filtering remains stable. \cite{Qian10423131} proved that CI-based distributed Kalman filtering converges in linear time-invariant systems by analyzing a coupled Riccati equation. Nevertheless, although there have been numerous works on the stability of the CI-based distributed Kalman filtering, several refinements are warranted: 1) In practice, the communication graph is often time-varying due to the energy-saving needs and sampling asynchrony of the nodes. However, only time-invariant graphs were discussed by \cite{BATTISTELLI2014707} and \cite{BATTISTELLI2016169,Liu8113573,CHEN2021109769,Li8845692,Wei8423433,Qian9599480,Qian10423131}. 2) Sampling asynchrony often causes the system to be periodically varying \cite{Fujimoto7587933,Qian10167860}. Notably, centralized Kalman filtering converges in periodic systems; however, the convergence of the CI-based distributed Kalman filtering still remains open in periodic systems.

Based on the discussion above, the contributions of this paper are summarized as follows:
\begin{itemize}
\item For general time-varying cases, it is demonstrated for the first time, that even if the network topology is not connected, the estimation error of one node is mean-square bounded, provided that it is uniformly observable with the assistance of those nodes that can uniformly reach it.
\item For periodic time-varying cases, by analyzing the monotonicity and boundedness of a Riccati-like difference equation, we prove that CI-based distributed Kalman filtering converges periodically independent of initial conditions.
\end{itemize}

\textbf{Notations:} ${\mathbb{R}}^r$ and ${\mathbb{R}}^{r\times s}$ denote the $r$ dimensional and $r\times s$ dimensional Euclidean spaces, respectively. $\oplus$ and $\otimes$ respectively stand for the direct sum and Kronecker product of matrices. $I$ stands for an identity matrix. $\mathrm{Tr}(\cdot)$, $\mathrm{Det}(\cdot)$ and $\|\cdot\|$ represent the trace, determinant and induced norm of a matrix, respectively. For $X,Y\in{\mathbb{R}}^{r\times r}$, $X>Y$ and $X\geq Y$ mean that $X-Y$ is symmetric positive definite and positive semi-definite, respectively. $|\cdot|$ denotes the cardinality of a finite set. $\mathrm{E}[\cdot]$ denotes the mathematical expectation. The composite of operators are denoted as $\mathcal{T}_1\circ \mathcal{T}_2(\cdot)\triangleq\mathcal{T}_1(\mathcal{T}_2(\cdot))$, and $\mathcal{T}_1^n(\cdot)\triangleq\mathcal{T}_1\circ\mathcal{T}_1\circ\cdots\circ \mathcal{T}_1(\cdot)$ ($n$ times).

\section{Problem Formulation}
\subsection{System description}
Consider the time-varying plant
\begin{equation}
x(k)=A(k-1)x(k-1)+w(k-1),
\end{equation}
where $x(k)\in\mathbb{R}^n$ is the system state, $\{w(k)\}_{k\geq0}$ is a Gaussian white noise process with covariance $Q(k)>0$, the initial state $x(0)$ is Gaussian with mean $x(0|0)$ and covariance $P(0|0)\geq0$. The plant (1) is observed by a network with $N$ sensor nodes, that is,
\begin{equation}
z_i(k)=C_i(k)x(k)+v_i(k),\ i=1,\cdots,N,
\end{equation}
where $z_i(k)\in\mathbb{R}^{m_i}$ is the measurement of node $i$, $\{v_i(k)\}_{k\geq1}$ is a Gaussian white noise process with covariance $R_i(k)>0$, and $\{v_i(k)\}_{k\geq1}$, $\{v_j(k)\}_{k\geq1}$, $\{w(k)\}_{k\geq0}$, and $x(0)$ are mutually uncorrelated ($i\neq j$). Nodes can communicate with each other, and their communication can be represented as a time-varying directed graph $\mathbb{G}(k)=(\mathbb{N},\mathbb{E}(k))$, where $\mathbb{N}=\{1,2,\cdots,N\}$ is the set of nodes and $\mathbb{E}(k)\subset\mathbb{N}\times \mathbb{N}$ denotes the set of edges among nodes at $k$th moment. An edge $(i,j)\in\mathbb{E}(k)$ indicates that node $i$ can send information to node $j$ at the $k$th moment (self-loop is permitted). Moreover, define in-neighbors and out-neighbors of node $i$ as $\mathbb{I}_i(k)\triangleq\{j|(j,i)\in\mathbb{E}(k)\}$ and $\mathbb{O}_i(k)\triangleq\{j|(i,j)\in\mathbb{E}(k)\}$, respectively.

To facilitate the subsequent analysis, it is necessary to introduce some concepts about the graph. For a time-invariant graph $\mathbb{G}=(\mathbb{N},\mathbb{E})$, if there is a path of edges $(i_0,i_1)$, $(i_1,i_2)$, $\cdots$, $(i_{s-1}, i_s)$ from node $i$ to another node $j$, then node $j$ is said to be reachable from node $i$, where $i_0=i$, $i_s=j$, and $(i_\iota,i_{\iota+1})\in\mathbb{E}$. The graph $\mathbb{G}$ is strongly connected if every node is reachable from every other node. For a time-varying graph $\mathbb{G}(k)$, it is jointly strongly connected if there exists $0\leq\nu<\infty$ such that $\bigcup^{k+\nu}_{i=k}\mathbb{G}(i)$ is strongly connected for all $k\geq0$. Meanwhile, if there exists $0\leq \mu<\infty$ such that node $i$ is reachable from node $j$ on the graph $\bigcup^{k+\mu}_{\iota=k}\mathbb{G}(\iota)$ for all $k\geq0$, then node $i$ is jointly reachable from node $j$. Moreover, we define $\mathbb{R}_i\triangleq\{j|\text{node $i$ is jointly reachable from $j$ on $\mathbb{G}(k)$}\}$ as the set of node that is jointly reachable to node $i$.

\subsection{CI-based distributed Kalman filtering}
The steps for CI-based distributed Kalman filtering are shown below.

Step 0 (Initialization). For each $i\in\mathbb{N}$, node $i$ inputs an initial estimate pair $(\hat{x}_i(0|0),\hat{P}_i(0|0))$, which should satisfy
\begin{equation}\begin{aligned}
&\mathrm{E}[\hat{x}_i(0|0)]=\mathrm{E}[x(0)],\\ &\hat{P}_i(0|0)\geq\mathrm{E}\big[\big(x(0)-\hat{x}_i(0|0)\big)\big(x(0)-\hat{x}_i(0|0)\big)^\mathrm{T}\big].
\end{aligned}\end{equation}

Step 1 (Local Estimation).
For each $i\in\mathbb{N}$, node $i$ calculates a local estimate pair $(x_{i}(k|k),P_{i}(k|k))$ utilizing Kalman filtering
\begin{equation}\begin{aligned}
x_{i}(k|k-1)=A(k-1)\hat{x}_{i}(k-1|k-1),
\end{aligned}\end{equation}
\begin{equation}\begin{aligned}
P_{i}(k|k-1)=&A(k-1)\hat{P}_{i}(k-1|k-1)A(k-1)^\mathrm{T}\\
&+Q(k-1),
\end{aligned}\end{equation}
\begin{equation}\begin{aligned}
K_{i}(k)=&P_{i}(k|k-1)C_i(k)^\mathrm{T}\\
&\times \big(C_i(k)P_{i}(k|k-1)C_i(k)^T+R_i(k)\big)^{-1},
\end{aligned}\end{equation}
\begin{equation}\begin{aligned}
x_{i}(k|k)=x_{i}(k|k-1)+K_{i}(k)\big(z_i(k)-z_i(k|k-1)\big),
\end{aligned}\end{equation}
\begin{equation}\begin{aligned}
P_{i}(k|k)=&P_{i}(k|k-1)-P_{i}(k|k-1)C_i(k)^\mathrm{T}\\
&\times\big(C_i(k)P_{i}(k|k-1)C_i(k)^\mathrm{T}+R_i(k)\big)^{-1}\\
&\times C_i(k)P_{i}(k|k-1).
\end{aligned}\end{equation}

Step 2 (Information Exchange). For each $i\in\mathbb{N}$, node $i$ receives in-neighbors' local estimate pairs $(x_{j}(k|k),P_{j}(k|k))$, $j\in\mathbb{I}_i(k)$.

Step 3 (Information Fusion). For each $i\in\mathbb{N}$, node $i$ combines the received estimate pairs by the CI fusion rule
\begin{equation}
\hat{x}_{i}(k|k)=\sum_{j\in\mathbb{N}}\pi_{ji}(k)\hat{P}_i(k|k)P_{j}(k|k)^{-1}x_{j}(k|k),
\label {eq:6}
\end{equation}
\begin{equation}\begin{aligned}
\hat{P}_i(k|k)=\big(\sum_{j\in\mathbb{N}}\pi_{ji}(k)P_{j}(k|k)^{-1}\big)^{-1},
\label {eq:7}
\end{aligned}\end{equation}
where $\pi_{ji}(k)\geq0$, $\sum_{j\in\mathbb{I}_i(k)}\pi_{ji}(k)=1$ and $\pi_{ji}(k)=0$ for $j\notin\mathbb{I}_i(k)$. Let $k+1\to k$ and return to Step 1.

\begin{remark}
Benefiting from the consistency of the CI fusion rule, the CI-based distributed Kalman filtering also ensures that the estimates of all nodes are consistent \cite{BATTISTELLI2014707}, i.e.,
\begin{equation}\begin{aligned}
&P_{i}(k|k)\geq\mathrm{E}\big[\big(x(k)-x_{i}(k|k)\big)\big(x(k)-x_{i}(k|k)\big)^\mathrm{T}\big],
\end{aligned}\end{equation}
\begin{equation}\begin{aligned}
&\hat{P}_i(k|k)\geq\mathrm{E}\big[\big(x(k)-\hat{x}_i(k|k)\big)\big(x(k)-\hat{x}_i(k|k)\big)^\mathrm{T}\big],
\end{aligned}\end{equation}
for $i\in\mathbb{N}$ and $k\geq0$.
\end{remark}

\subsection{Problems of interest}
Based on the discuss above, the problems to be solved are summarized as follows:
\begin{itemize}
  \item Analyze the stability of CI-based distributed Kalman filtering on the general time-varying system.
  \item Analyze the convergence of CI-based distributed Kalman filtering if the time-varying behavior is periodic.
\end{itemize}

\begin{remark}
This paper differs from the existing stability results on CI-based distributed Kalman filtering in the following aspects: i) \cite{BATTISTELLI2014707} and \cite{Li8845692,Wei8423433,Qian9599480,Qian10423131} only considered the time-invariant graph, whereas this paper addresses the time-varying graph. ii) \cite{BATTISTELLI2014707} and \cite{Li8845692,Wei8423433,Qian9599480,Qian10423131} assumed that the graph is strongly connected, whereas this paper does not require this assumption and proves that the stability of each node is only related to the nodes that can uniformly reach it (c.f. Theorem 1). iii) \cite{Qian10423131} proved by analyzing a coupled Riccati equation that CI-based distributed Kalman filtering converges in time-invariant systems, while this paper proves the convergence in the periodic time-varying systems using a completely different methodology (c.f. Theorems 2 and 3).
\end{remark}
\section{The General Time-Varying Case}
The following definitions and assumption are needed for our stability analysis.

\begin{definition}
For a time-varying matrix sequence $\{X(k)\}_{k\geq0}$, define
\begin{equation}
\Phi_X(i+j,i)\triangleq
\left\{ \begin{array}{l}
X(i+j-1)\cdots X(i),\ \mathrm{if}\ j>0,\\
I,\ \mathrm{if}\ j=0,\\
X(i+j)^{-1}\cdots X(i-1)^{-1},\ \mathrm{if}\ j<0.
\end{array} \right.\nonumber
\end{equation}
\end{definition}

\begin{definition}[\cite{rugh1996linear}]
The time-varying matrix pair $(Y(k),X(k))$ is uniformly observable if there exist finite constants $o>0$ and $\mu>0$ such that
\begin{equation}\begin{aligned}
\sum^{\mu-1}_{i=0}\Phi_X(i+k,k)^\mathrm{T}&Y(i+k)^\mathrm{T}Y(i+k)\\
&\times\Phi_X(i+k,k)\geq oI,\ \forall k\geq0.\nonumber
\end{aligned}\end{equation}
\end{definition}

\begin{assumption}
There exist finite constants $a_l>0$, $a_u>0$, $q_l>0$, $q_u>0$, $c_{u,i}>0$, $r_{l,i}>0$, $r_{u,i}>0$ and $\pi_{l,ji}>0$ such that
\begin{equation}\begin{aligned}
\left\{ \begin{array}{l}
a_lI\leq A(k)A(k)^\mathrm{T}\leq a_uI,\ q_lI\leq Q(k)\leq q_uI\\
C_i(k+1)C_i(k+1)^\mathrm{T}\leq c_{u,i}I\\
r_{i,l}I\leq R_i(k+1)\leq r_{i,u}I\\
\pi_{l,ji}\leq \pi_{ji}(k),\ \text{if}\ j\in\mathbb{I}_{i}(k)
\end{array} \right.\nonumber
\end{aligned}\end{equation}
for $k\geq 0$ and $i\in\mathbb{N}$.
\end{assumption}

\subsection{Stability Analysis}
The following theorem will provide a fixed-time stability result for CI-based distributed Kalman filtering. The proof of the theorem is roughly divided into two stages. In the first stage we establish a relationship between the upper bound of the covariance matrix and the observability Grammian. In the second stage, we derive that the observability Grammian of each node consists of those nodes that are uniformly reachable to it. Under this case, we can conclude that the stability of each node is only related to those uniformly reachable nodes and remains independent of the network's connectivity.

\begin{theorem}
Consider the plant (1), the sensor network (2), and CI-based distributed Kalman filtering (4)-(10). Let Assumption 1 hold. Define
\begin{equation}\begin{aligned}
\hat{P}(k|k)\triangleq \hat{P}_1(k|k)\oplus \hat{P}_2(k|k)\oplus\cdots\oplus \hat{P}_N(k|k),
\end{aligned}\end{equation}
\begin{equation}\begin{aligned}
&[C(k)]_{\mathbb{R}_i}\triangleq[C_{i_1}(k);C_{i_2}(k);\cdots;C_{i_{|\mathbb{R}_i|}}(k)],
\end{aligned}\end{equation}
\begin{equation}\begin{aligned}
\{i_1,i_2,\cdots,i_{|\mathbb{R}_i|}\}=\mathbb{R}_i.
\end{aligned}\end{equation}
If $([C(k)]_{\mathbb{R}_i},A(k))$ is uniformly observable, then for any $0\leq \hat{P}(0|0)<\infty I$, there exists finite constant $p_{u,i}>0$ and $K>0$ independent of $\hat{P}(0|0)$ such that
\begin{equation}\begin{aligned}
\mathrm{E}[\|x(k)-\hat{x}_i(k|k)\|^2]\leq\mathrm{Tr}(\hat{P}_i(k|k))\leq p_{u,i},\ \forall k\geq K.
\end{aligned}\end{equation}
\end{theorem}

\begin{proof}
By utilizing the  Sherman–Morrison–Woodbury formula \cite{horn2012matrix}, one can reformulate $P_{i}(k+1|k+1)$ as
\begin{equation}\begin{aligned}
P_{i}(k+1|k+1)=&\big(P^{-1}_{i}(k+1|k)+C_i(k+1)^\mathrm{T}\\
&\times R_i(k+1)^{-1}C_i(k+1)\big)^{-1}.
\end{aligned}\end{equation}
Substituting (17) into (10) yields
\begin{equation}\begin{aligned}
&\ \ \ \hat{P}^{-1}_i(k+1|k+1)\\
&=\sum_{j\in\mathbb{N}}\pi_{ji}(k+1)\big(A(k)\hat{P}_j(k|k)A(k)^\mathrm{T}+Q(k)\big)^{-1}\\
&\ \ \ \ \ \ \ \ +\pi_{ji}(k+1)C_j(k+1)^\mathrm{T}R_j(k+1)C_j(k+1).
\end{aligned}\end{equation}
Under Assumption 1, one can derive that
\begin{equation}\begin{aligned}
&\ \ \ \hat{P}^{-1}_i(k+1|k+1)\\
&\leq\sum_{j\in\mathbb{N}}\pi_{ji}(k+1)Q(k)^{-1}+\pi_{ji}(k+1)C_j(k+1)^\mathrm{T}\\
&\ \ \ \ \ \ \ \ \times R_j(k+1)C_j(k+1)\\
&\leq(\frac{1}{q_l}+r_{u}c_{u})I,
\end{aligned}\end{equation}
\begin{equation}\begin{aligned}
Q(k)\leq \frac{q_uq_l}{a_l(1+q_lr_{u}c_{u})}A(k)\hat{P}_{i}(k|k)A(k)^\mathrm{T},
\end{aligned}\end{equation}
where $r_u\triangleq\max\{r_{u,1},r_{u,2},\cdots,r_{u,N}\}$ and $c_u\triangleq\max\{c_{u,1},c_{u,2},\cdots,c_{u,N}\}$.
Then, it follows from (18) and (20) that
\begin{equation}\begin{aligned}
&\ \ \ \hat{P}^{-1}_i(k+1|k+1)\\
&\geq\sum_{j\in\mathbb{N}}\gamma \pi_{ji}(k+1)A(k)^\mathrm{-T}\hat{P}_{i}(k|k)^{-1}A(k)^{-1}\\
&\ \ \ \ \ \ \ \ +\pi_{ji}(k+1)C_j(k+1)^\mathrm{T}R_j(k+1)C_j(k+1),
\end{aligned}\end{equation}
where $\gamma=1/(1+\frac{q_uq_l}{a_l(1+q_lr_{u}c_{u})})=\frac{a_l(1+q_lr_{u}c_{u})}{q_uq_l+a_l(1+q_lr_{u}c_{u})}$.

Define the weighted adjoint matrix as
\begin{equation}\begin{aligned}
\Pi(k)\triangleq\begin{bmatrix}
{\pi_{11}(k)} & {\pi_{12}(k)} & {\cdots} & {\pi_{1N}(k)}\\
{\pi_{21}(k)} & {\pi_{22}(k)} & {\cdots} & {\pi_{2N}(k)}\\
{\vdots} & {\vdots} & {\ddots} & {\vdots}\\
{\pi_{N1}(k)} & {\pi_{N2}(k)} & {\cdots} & {\pi_{NN}(k)}
\end{bmatrix}.
\end{aligned}\end{equation}
Moreover, define $\Pi(k,k+\iota)\triangleq\Pi(k)\Pi(k+1)\times\cdots\times\Pi(k+\iota)$ and $\Pi_{ij}(k,k+\iota)$ as $(i,j)$th element of $\Pi(k,k+\iota)$. With these definitions, recursing the matrix inequality (21) yields
\begin{equation}\begin{aligned}
&\hat{P}_i(k+s|k+s)^{-1}\\
\geq&\frac{1}{r_u}\sum^{s}_{\iota=1}\gamma^{\iota-1}\sum_{j\in\mathbb{N}}\Pi_{ji}(k+s-\iota+1,k+s)\\
&\times\Phi_A(k+s-\iota+1,k+s)^\mathrm{T}C_j(k+s-\iota+1)^\mathrm{T}\\
&\times C_j(k+s-\iota+1)\Phi_A(k+s-\iota+1,k+s)\\
\geq&\frac{\gamma^{s-1}}{r_u}\Phi_A(k+s,k+1)^{-\mathrm{T}}\\
&\times\big(\sum^{s}_{\iota=1}\sum_{j\in\mathbb{N}}\Pi_{ji}(k+s-\iota+1,k+s)\\
&\times\Phi_A(k+s-\iota+1,k+1)^\mathrm{T}C_j(k+s-\iota+1)^\mathrm{T}\\
&\times C_j(k+s-\iota+1)\Phi_A(k+s-\iota+1,k+1)\big)\\
&\times\Phi_A(k+s,k+1)^{-1}.
\end{aligned}\end{equation}
Lemma A1 in Appendix shows that for any $k\geq1$ and $j\in\mathbb{R}_i$, there exists a path of length $0<s_i<\infty$ from node $j$ to node $i$ over the interval $[k,k+s_i]$. This implies that $\Pi_{ji}(k,k+s_i)\geq\beta_i$ for some $0<\beta_i<\infty$ and all $k\geq 1$ and $j\in\mathbb{R}_i$. Consequently, one can choose $0<s<\infty$ independent of $\hat{P}(0|0)$ such that
\begin{equation}\begin{aligned}
&\hat{P}^{-1}_i(k+s|k+s)\\
\geq&\beta_i\frac{\gamma^{s-1}}{r_u}\Phi_A(k+s,k+1)^{-\mathrm{T}}\\
&\times\big(\sum^{s}_{\iota=s_i+1}\Phi_A(k+s-\iota+1,k+1)^\mathrm{T}\\
&\times [C(k+s-\iota+1)]_{\mathbb{R}_i}^\mathrm{T}[C(k+s-\iota+1)]_{\mathbb{R}_i}\\
&\times\Phi_A(k+s-\iota+1,k+1)\big)\Phi_A(k+s,k+1)^{-1}\\
\geq&\alpha_i\beta_i\frac{\gamma^{s-1}}{r_ua_u^{s-1}},
\end{aligned}\end{equation}
where $0<\alpha_i<\infty$ is a constant independent of $\hat{P}(0|0)$ such that
\begin{equation}\begin{aligned}
\sum^{s}_{\iota=s_i+1}&\Phi_A(k+s-\iota+1,k+1)^\mathrm{T}\\
&\times [C(k+s-\iota+1)]_{\mathbb{R}_i}^\mathrm{T}[C(k+s-\iota+1)]_{\mathbb{R}_i}\\
&\times\Phi_A(k+s-\iota+1,k+1)\geq\alpha_iI
\end{aligned}\end{equation}
for $k\geq 1$. Since $([C(k)]_{\mathbb{R}_i},A(k))$ is uniformly observable, such a constant must exist provided that $0<s<\infty$ is sufficiently large. Moreover, it can be seen from (24) that the upper bound is independent of the initial condition. The proof is completed.
\end{proof}

%
\begin{corollary}
If the time-varying graph $\mathbb{G}(k)$ is jointly strongly connected and $([C(k)]_{\mathbb{N}},A(k))$ is uniformly observable, then (14) holds for $i\in\mathbb{N}$.
\end{corollary}

\begin{proof}
The time-varying graph $\mathbb{G}(k)$ being jointly strongly connected implies that $\mathbb{R}_i=\mathbb{N}$ for $i\in\mathbb{N}$. Then the corollary follows immediately from Theorem 1.
\end{proof}


\begin{remark}
Indeed, if $j$ does not belong to $\mathbb{R}_i$, then node $i$ is not able to get information from node $j$. In other words, Theorem 1 implies that the stability of each node is only relevant to those nodes that can uniformly affect it, which is consistent with our intuition.
\end{remark}
\section{The Periodic Time-Varying Case}
In this section we discuss the case where the system is periodic. This can be summarized as the following assumption.
\begin{assumption}
The system parameters satisfy
\begin{equation}\begin{aligned}
&A(k)=A(k+T),\ Q(k)=Q(k+T),\\
&C_i(k+1)=C_i(k+1+T),\\
&R_i(k+1)=R_i(k+1+T),\\
&\mathbb{G}(k+1)=\mathbb{G}(k+1+T),\\
&\pi_{ji}(k+1)=\pi_{ji}(k+1+T),\nonumber
\end{aligned}\end{equation}
for all $k\geq0$ and $i,j\in\mathbb{N}$, where $1\leq T<\infty$ is the period.
\end{assumption}


For ease of presentation, we define $\sigma_k$ as the reminder of $k\div T$ and $\mathbb{T}\triangleq\{0,1,\cdots,T-1\}$. Moreover, let us introduce several operators
\begin{equation}\begin{aligned}
\hat{\mathcal{L}}_\vartheta(X)\triangleq
\left\{ \begin{array}{l}
\mathcal{L}_{\vartheta-1}\circ\cdots\circ\mathcal{L}_{0}\circ\mathcal{L}_{T-1}\circ\cdots\circ\mathcal{L}_{\vartheta}(X),\\ \quad\quad\quad\quad\quad\quad\quad\quad\quad\quad\mathrm{if}\ \vartheta\in\mathbb{T}\setminus\{0\},\\
\mathcal{L}_{T-1}\circ\cdots\circ\mathcal{L}_{0}(X),\ \mathrm{if}\ \vartheta=0,
\end{array} \right.
\end{aligned}\end{equation}
\begin{equation}\begin{aligned}
\mathcal{L}_{\vartheta}(X)\triangleq\sum_{i\in\mathbb{N}}&S_i(\vartheta)^\mathrm{T}\big(\bar{A}(\vartheta)^{-\mathrm{T}}X\bar{A}(\vartheta)^{-1}+\bar{C}(\vartheta+1)^\mathrm{T}\\
&\times R(\vartheta+1)^{-1}\bar{C}(\vartheta+1)-\bar{A}(\vartheta)^{-\mathrm{T}}X\\
&\times(X+\bar{Q}(\vartheta))^{-1}X\bar{A}(\vartheta)^{-1}\big)S_i(\vartheta),
\end{aligned}\end{equation}
where
\begin{equation}\begin{aligned}
\bar{A}(k)\triangleq I\otimes A(k),\ \bar{Q}(k)=\bar{A}(k)^{\mathrm{T}}\big(I\otimes Q(k)\big)^{-1}\bar{A}(k),
\end{aligned}\end{equation}
\begin{equation}\begin{aligned}
S_i(k)\triangleq&\big(\bar{\Pi}(k+1)\otimes I\big)\\
&\times(0_{(i-1)n\times (i-1)n}\oplus I_{n}\oplus 0_{(N-i)n\times (N-i)n}),
\end{aligned}\end{equation}
\begin{equation}\begin{aligned}
\bar{C}(k)\triangleq C_1(k)\oplus C_2(k)\oplus\cdots\oplus C_N(k),
\end{aligned}\end{equation}
\begin{equation}\begin{aligned}
R(k)\triangleq R_1(k)\oplus R_2(k)\oplus\cdots\oplus R_N(k),
\end{aligned}\end{equation}
\begin{equation}\begin{aligned}
\bar{\Pi}(k)\triangleq\begin{bmatrix}
{\sqrt{\pi_{11}(k)}} & {\sqrt{\pi_{12}(k)}} & {\cdots} & {\sqrt{\pi_{1N}(k)}}\\
{\sqrt{\pi_{21}(k)}} & {\sqrt{\pi_{22}(k)}} & {\cdots} & {\sqrt{\pi_{2N}(k)}}\\
{\vdots} & {\vdots} & {\ddots} & {\vdots}\\
{\sqrt{\pi_{N1}(k)}} & {\sqrt{\pi_{N2}(k)}} & {\cdots} & {\sqrt{\pi_{NN}(k)}}
\end{bmatrix}.
\end{aligned}\end{equation}
\subsection{Convergence analysis}
The following theorem will show that CI-based distributed Kalman filtering converges in the periodic case. Considering the lengthy process of the proof, we present the idea of the proof roughly.

We first construct an intermediate difference equation. Second, we show that the difference equation converges periodically to $X_{\vartheta}$ ($\forall\vartheta\in\mathbb{T}$) for an initial value of $0$. Third, we show that the difference equation also converges periodically to $X_{\vartheta}$ for any initial value greater than or equal to $X_{\vartheta}$ in the sense of Loewner partial order. Fourth, since $0\leq X\leq X_{\vartheta}+X$ for any initial value $X\geq0$, it follows from the sandwich principle that the difference equation converges periodically to $X_{\vartheta}$ for any initial value. Finally, select specific initial conditions to equate the difference equation solution with the covariance of CI-based distributed Kalman filtering, thereby proving convergence.

\begin{theorem}
Consider the plant (1), sensor network (2), and CI-based distributed Kalman filtering (4)-(10). Let Assumptions 1 and 2 hold. If $([C(k)]_{\mathbb{R}_i},A(k))$ is uniformly observable for $i\in\mathbb{N}$, then for any $0\leq \hat{P}(0|0)<\infty I$ and $\vartheta\in\mathbb{T}$, there exists $0<\hat{P}_\vartheta<\infty I$ independent of $\hat{P}(0|0)$ such that
\begin{equation}\begin{aligned}
\lim_{j\to\infty}\hat{P}(&jT+\vartheta|jT+\vartheta)=\hat{P}_\vartheta,
\end{aligned}\end{equation}
where $\hat{P}(k|k)=\hat{P}_1(k|k)\oplus\hat{P}_2(k|k)\oplus\cdots\oplus\hat{P}_N(k|k)$. Moreover, $\hat{P}_\vartheta^{-1}$ is the unique symmetric positive definite solution to the matrix equation $X=\hat{\mathcal{L}}_\vartheta(X)$, $\vartheta\in\mathbb{T}$.
\end{theorem}

\begin{proof}
Let us define $X(k)=X_1(k)\oplus X_2(k)\oplus\cdots\oplus X_N(k)$, which is the solution to the intermediate matrix difference equation
\begin{equation}\begin{aligned}
X(k+1)=&\mathcal{L}_{\sigma_k}(X(k)).
\end{aligned}\end{equation}
Next, we will prove the theorem by analyzing the fixed point of the intermediate matrix difference equation (34).  Under Assumption 2, one can derive that
\begin{equation}\begin{aligned}
X((i+1)T+\vartheta)=\hat{\mathcal{L}}_\vartheta(X(iT+\vartheta)),\ \forall\vartheta\in\mathbb{T}.
\end{aligned}\end{equation}
To emphasize the effect of the initial condition on $X(k)$, define
\begin{equation}\begin{aligned}
X(k;k_0,X(k_0))=&\mathcal{L}_{\sigma_{k-1}}\circ\mathcal{L}_{\sigma_{k-2}}\circ\cdots\circ\mathcal{L}_{\sigma_{k_0}}(X(k_0))\nonumber
\end{aligned}\end{equation}
for $k>k_0$, and $X(k;k_0,X(k_0))=X(k_0)$ for $k=k_0$.

Introduce an auxiliary variable
\begin{equation}\begin{aligned}
\Theta\triangleq\sum_{i\in\mathbb{N}}&S_i(0)^\mathrm{T}\big((I\otimes Q(0))^{-1}+\bar{C}(1)^\mathrm{T}R(1)^{-1}\bar{C}(1)\big)S_i(0).\nonumber
\end{aligned}\end{equation}
Clearly, $\Theta>0$. Let $X(1)=\Theta$. It can be derived from Lemma A.2 in Appendix that
\begin{equation}\begin{aligned}
X(T+1;1,\Theta)=\mathcal{L}_{0}(X(T;1,\Theta))\leq\Theta=X(1).\nonumber
\end{aligned}\end{equation}
Then, utilizing Lemma A.3 in Appendix yields
\begin{equation}\begin{aligned}
X(2T+1;1,\Theta)=\hat{\mathcal{L}}_{1}(X(T+1;1,\Theta))\leq X(T+1;1,\Theta).\nonumber
\end{aligned}\end{equation}
Repeating the procedure shows that the sequence $\{X(iT+1;1,\Theta)\}_{i\geq0}$ is monotonically non-increasing. Consequently, one has
\begin{equation}\begin{aligned}
\lim_{i\to\infty}\hat{\mathcal{L}}^i_1(\Theta)=\lim_{i\to\infty}X(iT+1;1,\Theta)\leq\Theta.
\end{aligned}\end{equation}

Let $X(0)=0$, one can derive that
\begin{equation}\begin{aligned}
X(T;0,0)=\hat{\mathcal{L}}_{0}(X(0))\geq 0=X(0;0,0).
\end{aligned}\end{equation}
It follows from Lemma A3 in Appendix that
\begin{equation}\begin{aligned}
X(2T;0,0)=\hat{\mathcal{L}}_{0}(X(T;0,0))\geq \hat{\mathcal{L}}_{0}(0)=X(T,0,0).
\end{aligned}\end{equation}
Repeating the procedure shows that the sequence $\{X(iT;0,0)\}_{i\geq0}$ is monotonically non-decreasing. Moreover, one can know from the inequality $\Theta\geq\mathcal{L}_0(0)$ and Lemma A.3 in Appendix that
\begin{equation}\begin{aligned}
\lim_{i\to\infty}X(iT;0,0)&=\lim_{i\to\infty}\hat{\mathcal{L}}^i_0(0)\\
&=\lim_{i\to\infty}\mathcal{L}_{T-1}\circ\cdots\circ\mathcal{L}_{1}\circ\hat{\mathcal{L}}^{i-1}_1\circ\mathcal{L}_0(0)\\
&\leq\lim_{i\to\infty}\mathcal{L}_{T-1}\circ\cdots\circ\mathcal{L}_{1}\circ\hat{\mathcal{L}}^{i-1}_1(\Theta)\\
&\leq \mathcal{L}_{T-1}\circ\cdots\circ\mathcal{L}_{1}(\Theta).
\end{aligned}\end{equation}
This implies that the sequence $\{X(iT,0,0)\}_{i\geq0}$ is uniformly bounded from above. Consequently,
\begin{equation}\begin{aligned}
\lim_{i\to\infty}X(iT;0,0)=X_0,
\end{aligned}\end{equation}
where $0\leq X_0<\infty I$. Then, utilizing the equality (34) yields
\begin{equation}\begin{aligned}
\lim_{i\to\infty}X(iT+\vartheta;0,0)=X_\vartheta\ \text{for}\ \vartheta\in\mathbb{T},
\end{aligned}\end{equation}
where
\begin{equation}\begin{aligned}
X_\vartheta=
\left\{ \begin{array}{l}
\mathcal{L}_{\vartheta-1}(X_{\vartheta-1}),\ \mathrm{if}\ \vartheta\in\mathbb{T}\setminus\{0\},\\
\mathcal{L}_{T-1}(X_{T-1}),\ \mathrm{if}\ \vartheta=0.
\end{array} \right.
\end{aligned}\end{equation}

Let $X(0)=X_0+\Upsilon$, where $0\leq \Upsilon<\infty I$ is arbitrary. It is trivial that $X(0;0,X_0+\Upsilon)=X_0+\Upsilon\geq X_{\sigma_0}$. Suppose that $X(k-1;0,X_0+\Upsilon)\geq X_{\sigma_{k-1}}$. Then, it follows from (42) that
\begin{equation}\begin{aligned}
X(k;0,X_0+\Upsilon)&=\mathcal{L}_{\sigma_{k-1}}(X(k-1;0,X_0+\Upsilon))\\
&\geq\mathcal{L}_{\sigma_{k-1}}(X_{\sigma_{k-1}})=X_{\sigma_k}.
\end{aligned}\end{equation}
This induction implies that
\begin{equation}\begin{aligned}
X(iT+\vartheta;0,X_0+\Upsilon)\geq X_\vartheta,\ \forall i\geq0,\ \vartheta\in\mathbb{T}.
\end{aligned}\end{equation}
Note that the discrete-time Riccati equation has the reformulation
\begin{equation}\begin{aligned}
&AXA^\mathrm{T}+Q-AXC^\mathrm{T}(CXC^\mathrm{T}+R)^{-1}CXA^\mathrm{T}\\
=&(A-KC)X(A-KC)^\mathrm{T}+KRK^\mathrm{T}+Q
\end{aligned}\end{equation}
with $Q\geq0$, $R>0$ and $K=AXC^\mathrm{T}(CXC^\mathrm{T}+R)^{-1}$. Based on this reformulation, one can derive that
\begin{equation}\begin{aligned}
\mathcal{L}_{\vartheta}(X_\vartheta)=\mathcal{G}_{\vartheta}(X_\vartheta),\ \forall\vartheta\in\mathbb{T},
\end{aligned}\end{equation}
where
\begin{equation}\begin{aligned}
\mathcal{G}_{\vartheta}(X)\triangleq\mathcal{S}_\vartheta(X)+\Upsilon_\vartheta,
\end{aligned}\end{equation}
\begin{equation}\begin{aligned}
\mathcal{S}_{\vartheta}(X)\triangleq\sum_{i\in\mathbb{N}}&S_i(\vartheta)^\mathrm{T}\big(\bar{A}(\vartheta)^{-\mathrm{T}}-G(\vartheta)\big)X\\
&\times\big(\bar{A}(\vartheta)^{-\mathrm{T}}-G(\vartheta)\big)^\mathrm{T}S_i(\vartheta),
\end{aligned}\end{equation}
\begin{equation}\begin{aligned}
\Upsilon_{\vartheta}\triangleq\sum_{i\in\mathbb{N}}&S_i(\vartheta)^\mathrm{T}\big(G(\vartheta)\bar{Q}(\vartheta)G(\vartheta)^\mathrm{T}+\bar{C}(\vartheta+1)^\mathrm{T}\\
&\times R(\vartheta+1)^{-1}\bar{C}(\vartheta+1)\big)S_i(\vartheta),
\end{aligned}\end{equation}
\begin{equation}\begin{aligned}
G(\vartheta)\triangleq\bar{A}(\vartheta)^{-\mathrm{T}}X_\vartheta\big(X_\vartheta+\bar{Q}(\vartheta)\big)^{-1}.
\end{aligned}\end{equation}
Moreover, it follows from (42) and (46) that
\begin{equation}\begin{aligned}
X_\vartheta=\hat{\mathcal{G}}_\vartheta(X_{\vartheta}),\ \forall\vartheta\in\mathbb{T},
\end{aligned}\end{equation}
where
\begin{equation}\begin{aligned}
\hat{\mathcal{G}}_\vartheta(X)\triangleq
\left\{ \begin{array}{l}
\mathcal{G}_{\vartheta-1}\circ\cdots\circ\mathcal{G}_{0}\circ\mathcal{G}_{T-1}\circ\cdots\circ\mathcal{G}_{\vartheta}(X),\\ \quad\quad\quad\quad\quad\quad\quad\quad\quad\quad\mathrm{if}\ \vartheta\in\mathbb{T}\setminus\{0\},\\
\mathcal{G}_{T-1}\circ\cdots\circ\mathcal{G}_{0}(X),\ \mathrm{if}\ \vartheta=0.
\end{array} \right.
\end{aligned}\end{equation}
Then, some tedious manipulation yields (c.f. Appendix)
\begin{equation}\begin{aligned}
&\ \ \ X(iT+\vartheta;0,X_0+\Upsilon)-X_\vartheta\\
&\leq \hat{\mathcal{S}}_\vartheta^{i}(X(\vartheta;0,X_0+\Upsilon)-X_\vartheta),
\end{aligned}\end{equation}
where
\begin{equation}\begin{aligned}
\hat{\mathcal{S}}_\vartheta(X)\triangleq
\left\{ \begin{array}{l}
\mathcal{S}_{\vartheta-1}\circ\cdots\circ\mathcal{S}_{0}\circ\mathcal{S}_{T-1}\circ\cdots\circ\mathcal{S}_{\vartheta}(X),\\ \quad\quad\quad\quad\quad\quad\quad\quad\quad\quad\mathrm{if}\ \vartheta\in\mathbb{T}\setminus\{0\},\\
\mathcal{S}_{T-1}\circ\cdots\circ\mathcal{S}_{0}(X),\ \mathrm{if}\ \vartheta=0.
\end{array} \right.
\end{aligned}\end{equation}
Since $\Upsilon_{\vartheta}>0$ for $\vartheta\in\mathbb{T}$ (c.f. Appendix), one has
\begin{equation}\begin{aligned}
X_\vartheta=
\left\{ \begin{array}{l}
\mathcal{G}_{\vartheta-1}(X_{\vartheta-1})>\mathcal{S}_{\vartheta-1}(X_{\vartheta-1}),\ \mathrm{if}\ \vartheta\in\mathbb{T}\setminus\{0\},\\
\mathcal{G}_{T-1}(X_{T-1})>\mathcal{S}_{T-1}(X_{T-1}),\ \mathrm{if}\ \vartheta=0.
\end{array} \right.\nonumber
\end{aligned}\end{equation}
This implies that
\begin{equation}\begin{aligned}
X_\vartheta&>\mathcal{S}_{\vartheta-1}(X_{\vartheta-1})\geq\cdots\\
&\geq\mathcal{S}_{\vartheta-1}\circ\cdots\circ\mathcal{S}_{0}\circ\mathcal{S}_{r-1}\circ\cdots\circ\mathcal{S}_{\vartheta}(X_\vartheta)\\
&=\hat{\mathcal{S}}_{\vartheta}(X_\vartheta),\ \forall\vartheta\in\mathbb{T}\setminus\{0\}
\end{aligned}\end{equation}
and
\begin{equation}\begin{aligned}
X_0&>\mathcal{S}_{T-1}(X_{T-1})\geq\cdots\\
&\geq\mathcal{S}_{T-1}\circ\cdots\circ\mathcal{S}_{0}(X_0)=\hat{\mathcal{S}}_0(X_0).
\end{aligned}\end{equation}
Based on (55) and (56), one can derive from the proposition A.3 in \cite{Fan992124} that
\begin{equation}\begin{aligned}
&\lim_{i\to\infty}X(iT+\vartheta;0,X_0+\Upsilon)-X_\vartheta\\
\leq&\lim_{i\to\infty}\hat{\mathcal{S}}_\vartheta^{i}(X(\vartheta;0,X_0+\Upsilon)-X_\vartheta)=0,\ \forall\vartheta\in\mathbb{T}.
\end{aligned}\end{equation}
Combining the inequalities (44) and (57) yields
\begin{equation}\begin{aligned}
\lim_{i\to\infty}X(iT+\vartheta;0,X_0+\Upsilon)=X_\vartheta,\ \forall \vartheta\in\mathbb{T}.
\end{aligned}\end{equation}
Since $0\leq\Upsilon<\infty I$ is arbitrary, one can conclude that
\begin{equation}\begin{aligned}
\lim_{i\to\infty}X(iT+\vartheta;0,X(0))=X_\vartheta,\ \forall \vartheta\in\mathbb{T},\ X_0\leq X(0)<\infty I.
\end{aligned}\end{equation}

For any $0\leq X(0)<\infty I$, the inequality $0\leq X(0)\leq X(0)+X_0$ is trivial. In this case, it can be derived from Lemma A4 in Appendix and the convergence results (41) and (59) that
\begin{equation}\begin{aligned}
\lim_{i\to\infty}X(iT+\vartheta;0,X(0))=X_\vartheta,\ \forall\ \vartheta\in\mathbb{T},\ 0\leq X(0)<\infty I.
\end{aligned}\end{equation}
Based on (41), one knows that the limit (60) is independent of the initial condition $X(0)\geq0$. Thus, $X_0$ is the unique symmetric positive definite solution to $X=\hat{\mathcal{L}}_0(X)$. Additionally, for $\vartheta\in\{1,2,\cdots,T-1\}$, one can decomposition $\hat{\mathcal{L}}_\vartheta^i$ as
\begin{equation}\begin{aligned}
\hat{\mathcal{L}}_\vartheta^i=\mathcal{L}_{\vartheta-1}\circ\cdots\circ\mathcal{L}_{0}\circ\hat{\mathcal{L}}_0^{i-1}\circ\mathcal{L}_{T-1}\circ\cdots\circ\mathcal{L}_{\vartheta}.
\end{aligned}\end{equation}
This implies that for any $0\leq X(\vartheta)<\infty I$,
\begin{equation}\begin{aligned}
&\lim_{i\to\infty}X(iT+\vartheta;\vartheta,X(\vartheta))\\
&=\lim_{i\to\infty}\hat{\mathcal{L}}_\vartheta^i(X(\vartheta))=\mathcal{L}_{\vartheta-1}\circ\cdots\circ\mathcal{L}_{0}(X_0)=X_\vartheta.
\end{aligned}\end{equation}
Thus, $X_\vartheta$ is the unique symmetric positive definite solution to $X=\hat{\mathcal{L}}_\vartheta(X)$ for $\vartheta\in\mathbb{T}$.

For any $0\leq \hat{P}(0|0)<\infty I$, one has $\hat{P}(1|1)>0$. Thus, if we set $X(1)^{-1}=\hat{P}(1|1)$, then Lemma A.2 implies that $X(k;1,X(1))^{-1}=\hat{P}(k|k)$ for $k\geq1$. Consequently, one has
\begin{equation}\begin{aligned}
&\lim_{i\to\infty}\hat{P}(iT+\vartheta|iT+\vartheta)^{-1}\\
=&\lim_{i\to\infty}\hat{\mathcal{L}}^i_\vartheta(X(\vartheta;0,X(0)))=X_\vartheta.
\end{aligned}\end{equation}
The proof is completed.
\end{proof}

%
\begin{remark}
Sometimes the periods of the graph $\mathbb{G}(k)$, plant (1), and sensor measurements (2) may be different. In this case, we can choose their least common multiple as the period of the whole system. Additionally, although we only prove that $\hat{P}_{i}(k|k)$ converges periodically, it is not difficult to conclude from (4)-(10) that the estimator parameters $K_i(k|k)$, $P_{i}(k|k-1)$ and $P_i(k|k)$ will also converge periodically.
\end{remark}

\subsection{Closed-Loop performance analysis}
Let us define
\begin{equation}
W_{i,j}(k)\triangleq
\left\{ \begin{array}{l}
\pi_{ji}(k)\hat{P}_i(k|k)P_{j}(k|k)^{-1},\ \mathrm{if}\ j\in\mathbb{I}_i(k),\\
0,\ \mathrm{if}\ j\notin\mathbb{I}_i(k).
\end{array} \right.
\end{equation}
Then, one can derive the recursive form of the local estimation error $e_{i}(k|k)\triangleq x(k)-x_{i}(k|k)$ as
\begin{equation}\begin{aligned}
e_{i}(k|k)&=\big(A(k-1)-K_{i}(k)C_i(k)A(k-1)\big)\\
&\ \ \ \times\sum_{j\in\mathbb{N}}W_{i,j}(k-1)e_j(k-1|k-1)\\
&\ \ \ +\big(I-K_{i}(k)C_i(k)\big)w(k-1)-K_{i}(k)v_i(k).\nonumber
\end{aligned}\end{equation}
Then, the global estimate error $e(k|k)\triangleq[e_{1}(k|k);e_{2}(k|k);\cdots;e_{N}(k|k)]$ can be recursively expressed by
\begin{equation}\begin{aligned}
e(k|k)&=F(k-1)e(k-1|k-1)+H(k)w(k-1)\\
&\ \ \ +K(k)v(k),
\end{aligned}\end{equation}
where
\begin{equation}\begin{aligned}
F(k)\triangleq&\begin{bmatrix}
{F_{1,1}(k)} & {F_{1,2}(k)} & {\cdots} & {F_{1,N}(k)}\\
{F_{2,1}(k)} & {F_{2,2}(k)} & {\cdots} & {F_{2,N}(k)}\\
{\vdots} & {\vdots} & {\vdots} & {\vdots}\\
{F_{N,1}(k)} & {F_{N,2}(k)} & {\cdots} & {F_{N,N}(k)}\\
\end{bmatrix},
\end{aligned}\end{equation}
\begin{equation}\begin{aligned}
K(k)\triangleq K_{1}(k)\oplus K_{2}(k)\cdots\oplus K_{N}(k),
\end{aligned}\end{equation}
\begin{equation}\begin{aligned}
H(k)\triangleq &\begin{bmatrix}
{I-K_{1}(k)C_1(k)}\\
{I-K_{2}(k)C_2(k)}\\
{\vdots}\\
{I-K_{N}(k)C_N(k)}\\
\end{bmatrix},
\end{aligned}\end{equation}
\begin{equation}\begin{aligned}
F_{i,j}(k)\triangleq\big(A(k)-K_{i}(k+1)C_i(k+1)A(k)\big)W_{i,j}(k).
\end{aligned}\end{equation}
Theorem 2 has proved that the CI-based distributed Kalman filtering converges in the periodic system. Thus, for $\vartheta\in\mathbb{T}$, we define the steady states
\begin{equation}\begin{aligned}
&\lim_{k\to\infty}P_{i}(k)=P_{i}^{[\vartheta]},\ \lim_{k\to\infty}\hat{P}_i(k)=\hat{P}_i^{[\vartheta]},\\ &\lim_{k\to\infty}P_{i}(k|k-1)=\bar{P}_{i}^{[\vartheta]},\ \lim_{k\to\infty}F(k)=F^{[\vartheta]},\\
&\lim_{k\to\infty}H(k)=H^{[\vartheta]},\ \lim_{k\to\infty}K(k)=K^{[\vartheta]},\\
&\lim_{k\to\infty}K_{i}(k)=K_{i}^{[\vartheta]},\ \lim_{k\to\infty}W_{i,j}(k)=W_{i,j}^{[\vartheta]},
\end{aligned}\end{equation}
where $k\in\{\kappa T+\vartheta|\kappa=0,1,2,\cdots\}$.

The following theorem will present the closed-loop performance of the estimation error system (65).

\begin{theorem}
Consider the plant (1), the sensor network (2), and CI-based distributed Kalman filtering (4)-(10). Let Assumptions 1 and 2 hold. If $([C(k)]_{\mathbb{R}_i},A(k))$ is uniformly observable for $i\in\mathbb{N}$, then
\begin{equation}\begin{aligned}
\rho(F^{[\vartheta-1]}\cdots F^{[0]}F^{[T-1]}\cdots F^{[\vartheta]})<1,\ \forall\vartheta\in\mathbb{T}\setminus \{0\},
\end{aligned}\end{equation}
\begin{equation}\begin{aligned}
\rho(F^{[T-1]}F^{[T-2]}\cdots F^{[0]})<1.
\end{aligned}\end{equation}
\end{theorem}

\begin{proof}
Consider an intermediate difference equation
\begin{equation}\begin{aligned}
X(k)&=F^{[\sigma_{k-1}]}X(k-1)(F^{[\sigma_{k-1}]})^\mathrm{T}+V^{[\sigma_k]},
\end{aligned}\end{equation}
where $X(0)=0$, $0<\alpha<1$, and
\begin{equation}\begin{aligned}
V^{[\vartheta]}&\triangleq\alpha H^{[\vartheta]}Q(\vartheta)(H^{[\vartheta]})^\mathrm{T}+(1-\alpha)\Psi^{[\vartheta]}\\
&\ \ \ +K^{[\vartheta]}R(\vartheta)(K^{[\vartheta]})^\mathrm{T}.
\end{aligned}\end{equation}
\begin{equation}\begin{aligned}
\Psi_{i,i}^{[\vartheta]}\triangleq\big(I-K_{i}^{[\vartheta]}C_i(\vartheta)\big)Q(\vartheta)\big(I-K_{i}^{[\vartheta]}C_i(\vartheta)\big)^\mathrm{T},
\end{aligned}\end{equation}
\begin{equation}\begin{aligned}
\Psi^{[\vartheta]}\triangleq\Psi_{1,1}^{[\vartheta]}\oplus\Psi_{2,2}^{[\vartheta]}\oplus\cdots\oplus\Psi_{N,N}^{[\vartheta]}.
\end{aligned}\end{equation}
Then, one can derive that
\begin{align}
X(kT)&=F^{[T-1]}\cdots F^{[0]}X((k-1)T)(F^{[T-1]}\cdots F^{[0]})^\mathrm{T}\nonumber\\
&\ \ \ +V^{[0]}+\sum^{T-1}_{i=1}F^{[T-1]}\cdots F^{[T-i]}V^{[T-i]}\nonumber\\
&\ \ \ \times(F^{[T-1]}\cdots F^{[T-i]})^\mathrm{T}.
\end{align}
Note that
\begin{align}
&\ \ \ \mathrm{Det}(I-K_{i}^{[\vartheta]}C_i(\vartheta))\nonumber\\
&=\mathrm{Det}\big(I-\big(C_i(\vartheta)\bar{P}_{i}^{[\vartheta]}C_i(\vartheta)^\mathrm{T}+R_i(\vartheta)\big)^{-1}C_i(\vartheta)\bar{P}_{i}^{[\vartheta]}C_i(\vartheta)^\mathrm{T}\big)\nonumber\\
&=\mathrm{Det}\big(\big(C_i(\vartheta)\bar{P}_{i}^{[\vartheta]}C_i(\vartheta)^\mathrm{T}+R_i(\vartheta)\big)^{-1}R_i(\vartheta)\big)>0.
\end{align}
Thus, $I-K_{i}^{[\vartheta]}C_i^{[\vartheta]}$ is invertible for $\vartheta\in\mathbb{T}$. This implies that $V^{[0]}>0$. Meanwhile, one can verify by some tedious manipulation (c.f. Appendix) that
\begin{equation}\begin{aligned}
X(k)\leq \sum_{i\in\mathbb{N}}\mathrm{Tr}(P_{i}^{[\vartheta]}) I,\ \forall k\geq0.
\end{aligned}\end{equation}
Consequently, it follows from the Lyapunov theory that (72) holds. Moreover, (71) follows immediately from (72) and the equality $\rho(XY)=\rho(YX)$ (for $X$ and $Y$ with appropriate dimensions). The proof is completed.
\end{proof}

\begin{remark}
The convergence result provides theoretical supports for the implementation of the steady-state version of CI-based distributed Kalman filtering. After a sufficiently long duration, all estimator parameters converge to their steady states periodically. In this scenario, each node needs to store only $T$ parameter sets to run CI-based distributed Kalman filtering. Consequently, exchanging $P_{i}(k|k)$ is no longer necessary. This effectively reduces the computation and communication burden of the network.
\end{remark}

\section{Simulations}
In this section, we consider a target tracking system whose dynamic process can be represented as
\begin{equation}\begin{aligned}
x(k+1)=A(k)x(k)+\omega(k),\nonumber
\end{aligned}\end{equation}
where
\begin{equation}\begin{aligned}
&A(k)=\begin{bmatrix}
{1} & {\Delta t_k} & {0} & {0}\\
{0} & {1} & {0} & {0}\\
{0} & {0} & {1} & {\Delta t_k}\\
{0} & {0} & {0} & {1}
\end{bmatrix},\\
&Q(k)=S_w\begin{bmatrix}
{\frac{\Delta t_k^4}{3}} & {\frac{\Delta t_k^3}{2}} & {0} & {0}\\
{\frac{\Delta t_k^3}{2}} & {\Delta t_k^2} & {0} & {0}\\
{0} & {0} & {\frac{\Delta t_k^4}{3}} & {\frac{\Delta t_k^3}{2}}\\
{0} & {0} & {\frac{\Delta t_k^3}{2}} & {\Delta t_k^2}
\end{bmatrix},\nonumber
\end{aligned}\end{equation}
$\Delta t_k$ is the sampling period, and $S_w=0.2$ is the power spectral density. A network with $N=10$ nodes are utilized to observe the target, and their measurement equations are given by
\begin{equation}\begin{aligned}
z_i(k+1)=C_ix(k+1)+v_i(k+1),\nonumber
\end{aligned}\end{equation}
where the measurement matrices are given by
\begin{equation}\begin{aligned}
\left\{ \begin{array}{l}
C_1=C_3=C_6=\begin{bmatrix}
{1} & {0} & {0} & {0}\\
{0} & {0} & {1} & {0}
\end{bmatrix},\\
C_2=C_4=C_5=\begin{bmatrix}
{0} & {1} & {0} & {0}\\
{0} & {0} & {0} & {1}
\end{bmatrix},\\
C_7=C_8=\begin{bmatrix}
{0} & {1} & {0} & {0}
\end{bmatrix},\\
C_9=C_{10}=\begin{bmatrix}
{0} & {0} & {0} & {1}
\end{bmatrix}.
\end{array} \right.
\end{aligned}\end{equation}
The topology of the network are shown in Fig. 2, where the state of the edge $(i,j)$ is determined by
\begin{equation}\begin{aligned}
(i,j)\ \text{is}
\left\{ \begin{array}{l}
\text{activated at the $k$-th moment if}\ \cos(a_ik)\geq-0.5,\\
\text{closed at the $k$-th moment if}\ \cos(a_ik)<-0.5,\\
\end{array} \right.\nonumber
\end{aligned}\end{equation}
and $a_i$ is a constant. Moreover, the noise covariances $R_i(k)$ is given by
\begin{equation}\begin{aligned}
\left\{ \begin{array}{l}
R_1=R_3=R_6=I,\\
R_2=R_4=R_5=0.1I,\\
R_7=R_8=R_9=R_{10}=0.1I.
\end{array} \right.\nonumber
\end{aligned}\end{equation}
\begin{figure}
      \centering
      \includegraphics[scale=0.7]{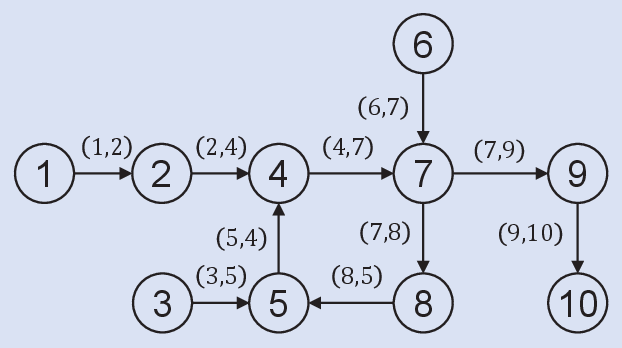}
      \caption{The topology of the sensor network.}
\end{figure}

We first verify our stability result in the general time-varying case. In this case, the parameters $\Delta t_k$ and $a_i$ are set as
\begin{equation}\begin{aligned}
\Delta t_k=1+0.1\sin(k),\ a_i=i.\nonumber
\end{aligned}\end{equation}
Then, one can verify that
\begin{equation}\begin{aligned}
&\mathbb{R}_1=\{1\},\ \mathbb{R}_2=\{1,2\},\ \mathbb{R}_3=\{3\},\ \mathbb{R}_6=\{6\},\\
&\mathbb{R}_4=\mathbb{R}_5=\mathbb{R}_7=\mathbb{R}_8=\mathbb{N}\setminus\{9,10\},\\
&\mathbb{R}_9=\mathbb{N}\setminus\{10\},\ \mathbb{R}_{10}=\mathbb{N}.
\end{aligned}\end{equation}
Clearly, the network is not jointly strongly connected (and not strongly connected). However, one can verify that $([C]_{\mathbb{R}_i},A(k))$ is uniformly observable for $i\in\mathbb{N}$. Then, we can see from Fig. 3 that the MSE of each node is bounded. Subsequently, we change $C_1$ in (80) to $C_1=[0,1,0,0;0,0,0,1]$, in which case $([C]_{\mathbb{R}_1},A(k))$ and $([C]_{\mathbb{R}_2},A(k))$ are no longer uniformly observable (but the other nodes remain). Figure 5 shows the MSE after $C_1$ being replaced, from which it can be seen that although nodes 1 and 2 diverge, the other nodes are still stable. These simulations are consistent with our stability result in Theorem 1.
\begin{figure}
      \centering
      \includegraphics[scale=0.6]{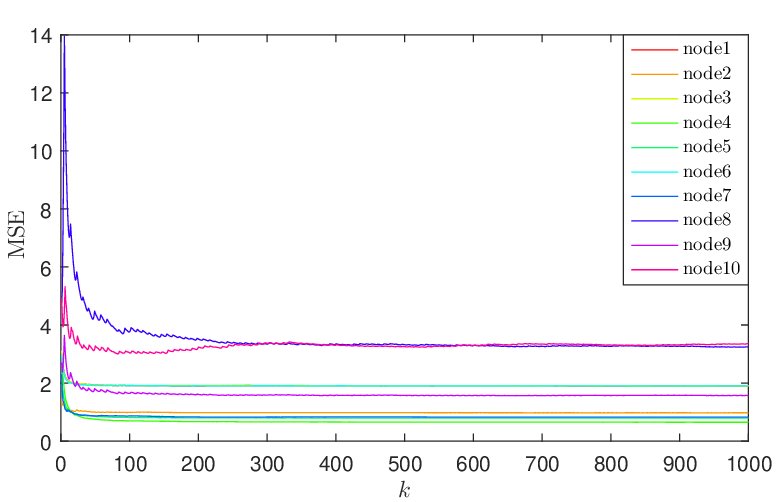}
      \caption{The MSE of CI-based distributed Kalman filtering in the general time-varying system.}
\end{figure}
\begin{figure}
      \centering
      \includegraphics[scale=0.6]{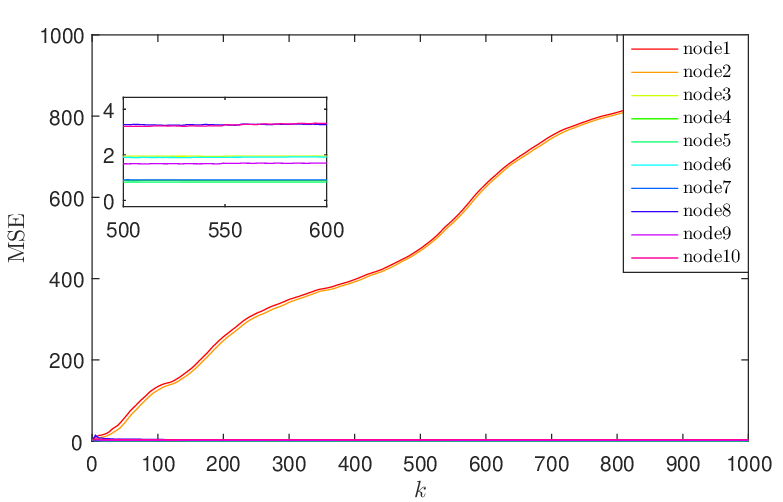}
      \caption{The MSE of CI-based distributed Kalman filtering in the general time-varying system ($C_1$ in (80) is replaced by $C_1=[0,1,0,0;0,0,0,1]$).}
\end{figure}

We then verify our convergence result in the periodic time-varying system. In this case, the parameters $\Delta t_k$ and $a_i$ are set as
\begin{equation}\begin{aligned}
\Delta t_k=1+0.1\cos(\pi k),\nonumber
\end{aligned}\end{equation}
\begin{equation}\begin{aligned}
a_1=a_2=a_3=a_4=a_5=\pi,\nonumber
\end{aligned}\end{equation}
\begin{equation}\begin{aligned}
a_6=a_7=a_8=a_9=a_{10}=\frac{\pi}{2}.\nonumber
\end{aligned}\end{equation}
The period of the whole system is $T=4$. Moreover, $\mathbb{R}_i$ is also given by (81). Fig. 4 shows the trace of the error covariance $\hat{P}_i(k|k)$, from which one can see that $\mathrm{Tr}(\hat{P}_i(k|k))$ converges with period 4, conforming to the convergence result in Theorem 2.
\begin{figure}
      \centering
      \includegraphics[scale=0.6]{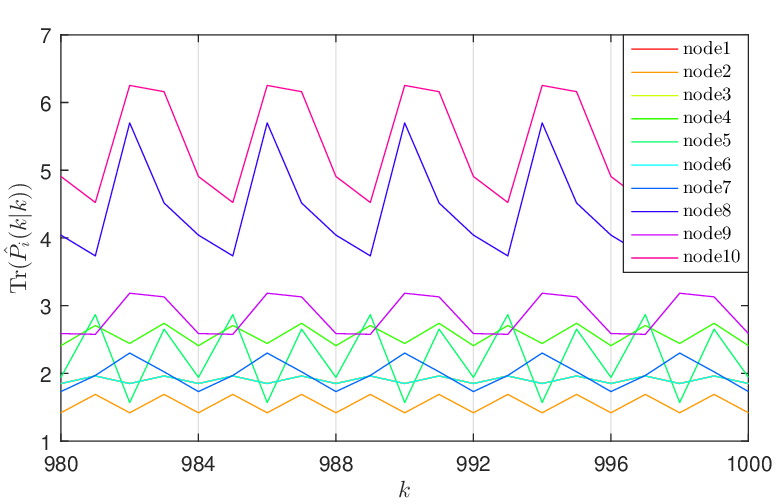}
      \caption{The trace of the error covariance of CI-based distributed Kalman filtering in the periodic time-varying system.}
\end{figure}

\section{Conclusion}
This paper investigated the stability of CI-based distributed Kalman filtering over the time-varying directed graph. We showed that the stability of one node is only related to the observability of those nodes that can reach it uniformly. Particularly, when the directed graph is strongly connected, our stability condition reduces to the existing stability condition, namely, that all nodes are collectively observable. Based on the stability result, we proved that CI-based distributed Kalman filtering converges periodically for any initial condition if the system is periodically time-varying. This implies that we can implement the steady-state version of CI-based distributed Kalman filtering to reduce the communication burden.

\section*{Appendix}
\setcounter{lemma}{0}
\renewcommand{\thelemma}{A.\arabic{lemma}}
\begin{lemma}
If $j\in\mathbb{R}_i$, then there exist $0<s_i<\infty$ and a path $(j_0,j_1)$, $(j_1,j_2)$, $\cdots$, $(j_{s_i},j_{s_i+1})$ from node $j$ to node $i$ over the interval $[k,k+s_i]$ for all $k\geq t$, where $j_0=j$, $j_{s_i+1}=i$ and $(j_\iota,j_{\iota+1})\in\mathbb{E}(k+\iota)$.
\end{lemma}

\begin{proof}
For a graph with $N$ nodes, we know that there are at most $M=\sum^{N-2}_{\iota=0}P(N,\iota)$ unduplicated paths from node $j$ to node $i$, where $i\neq j$ and $P(N,\iota)=N!/(N-\iota)!$ is the permutation. We denote them as $\mathcal{P}=\{\mathcal{P}_1,\mathcal{P}_2,\cdots,\mathcal{P}_M\}$. Here, $\mathcal{P}_\iota$ is an unduplicated path implies that a node is not visited twice on the path.

Define $\mathbb{G}(k,k+\kappa)\triangleq\bigcup^{k+\kappa}_{\iota=k}\mathbb{G}(\iota)$. By the definition of $\mathbb{R}_i$ we know that there exists an unduplicated path $\mathcal{P}_{\varsigma_{k+1}}\in\mathcal{P}$ from $j$ to $i$ on the graph $\mathbb{G}(k+1,k+\mu)=\bigcup^{k+\mu}_{\iota=k+1}\mathbb{G}(\iota)$ for $k\geq0$. Consider the joint graph $\mathbb{G}(k+1,k+M(N-1)\mu)=\bigcup^{M(N-1)}_{\iota=1}\mathbb{G}(k+(\iota-1)\mu+1,k+\iota\mu)$. Since $|\mathcal{P}|=M$, there must exist $\mathcal{P}^\prime\in\mathcal{P}$ and $1\leq\iota_1,\iota_2,\cdots,\iota_{N-1}\leq M(N-1)$ such that $\mathcal{P}^\prime=\mathcal{P}_{k+(\iota_\kappa-1)\mu+1}$ for all $\kappa=1,2,\cdots,N-1$. Since the path $\mathcal{P}^\prime$ occurs (at least) $N-1$ times and the length of any unduplicated path is less than or equal to $N-1$, the assertion holds.
\end{proof}

\begin{lemma}
Consider the matrix difference equation (34). If $X(t)>0$ and $\hat{P}(t|t)=X(t)^{-1}$, then $X(k)>0$ and $X(k)^{-1}=\hat{P}(k|k)$ for $k\geq t$. Moreover, for $k\geq t$, the difference equation (34) can be reformulated as
\begin{equation}\begin{aligned}
X(k+1)=\sum_{i\in\mathbb{N}}&S_i(\sigma_k)^\mathrm{T}\big((\bar{A}(\sigma_k)X(k)^{-1}\bar{A}(\sigma_k)^\mathrm{T}\\
&+I\otimes Q(\sigma_k))^{-1}+\bar{C}(\sigma_{k+1})^\mathrm{T}\\
&\times R(\sigma_{k+1})^{-1}\bar{C}(\sigma_{k+1})\big)S_i(\sigma_k).
\end{aligned}\end{equation}
\end{lemma}

\begin{proof}
This lemma will be proved by an induction. At the initial moment, it is trivial by the assumption that $X(t)>0$ and $\hat{P}(t|t)=X(t)^{-1}$. Suppose that $X(k)>0$ and $\hat{P}(k|k)=X(k)^{-1}$ for some $k>t$. Then, one can derive from the Sherman–Morrison–Woodbury formula that (82) holds.
Since $Q(k)$ is invertible, one has $X(k+1)>0$. Then, a comparison of (18) and (82) shows that $X(k+1)^{-1}=\hat{P}(k+1|k+1)$. The proof is completed.
\end{proof}

\begin{lemma}
Consider the operator $\hat{\mathcal{L}}_\vartheta(\cdot)$ defined in (26), then for any $\vartheta\in\mathbb{T}$, $\hat{\mathcal{L}}_\vartheta(X)\geq\hat{\mathcal{L}}_\vartheta(Y)$ if $X\geq Y$.
\end{lemma}

\begin{proof}
If $X\geq Y$, it follows from the monotonicity of the discrete-time Riccati equation \cite{Sinopoli1333199} that
\begin{equation}\begin{aligned}
\mathcal{L}_\vartheta(X)\geq\mathcal{L}_\vartheta(Y),\ \forall \vartheta\in\{0,1,\cdots,T-1\}.\nonumber
\end{aligned}\end{equation}
Substituting this inequality into the definition (26) completes the proof.
\end{proof}

\begin{lemma}
Consider the operator $\hat{\mathcal{L}}_\vartheta(\cdot)$ defined in (26). Let $Y_{\vartheta,1}(k+1)=\hat{\mathcal{L}}_\vartheta(Y_{\vartheta,1}(k))$ and $Y_{\vartheta,2}(k+1)=\hat{\mathcal{L}}_\vartheta(Y_{\vartheta,2}(k))$ with $Y_{\vartheta,1}(0)\geq Y_{\vartheta,2}(0)$, where $\vartheta\in\{0,1,\cdots,T-1\}$. Then, $Y_{\vartheta,1}(k)\geq Y_{\vartheta,2}(k)$ for all $k\geq0$ and $\vartheta\in\mathbb{T}$.
\end{lemma}

\begin{proof}
This lemma will be proved by an induction. It is trivial that $Y_{\vartheta,1}(0)\geq Y_{\vartheta,2}(0)$ for $\vartheta\in\mathbb{T}$. Suppose that $Y_{\vartheta,1}(k-1)\geq Y_{\vartheta,2}(k-1)$ for $\vartheta\in\mathbb{T}$. Then, one can derive from Lemma A.3 that
\begin{equation}\begin{aligned}
Y_{\vartheta,1}(k)&=\hat{\mathcal{L}}_\vartheta(Y_{\vartheta,1}(k-1))\\
&\geq \hat{\mathcal{L}}_\vartheta(Y_{\vartheta,2}(k-1))=Y_{\vartheta,2}(k).
\end{aligned}\end{equation}
The proof is completed.
\end{proof}

\begin{lemma}
Let $0\leq w_{i}\leq 1$ and $\sum^{N}_{i=1}w_i=1$. Let
\begin{equation}\begin{aligned}
X=\begin{bmatrix}
{X_{1,1}} & {X_{1,2}} & {\cdots} & {X_{1,N}}\\
{X_{2,1}} & {X_{2,2}} & {\cdots} & {X_{2,N}}\\
{\vdots} & {\vdots} & {\ddots} & {\vdots}\\
{X_{N,1}} & {X_{N,2}} & {\cdots} & {X_{N,N}}
\end{bmatrix}\geq0,\nonumber
\end{aligned}\end{equation}
$Y_{i}>0$ and $Y_{i}\geq X_{i,i}$ for $i=1,2,\cdots,N$. Let $Y=(\sum^N_{i=1}w_iY_i^{-1})^{-1}$. Then, one has
\begin{equation}\begin{aligned}
Y\geq Y(\sum^N_{i=1}\sum^N_{j=1}w_iw_jY_i^{-1}X_{i,j}Y_j^{-1})Y.
\end{aligned}\end{equation}
\end{lemma}

\begin{proof}
According to Theorem 7.7.9 in \cite{horn2012matrix}, one can reformulate $X_{i,j}$ as
\begin{equation}\begin{aligned}
X_{i,j}=\sqrt{X_{i,i}}Z_{i,j}\sqrt{X_{j,j}},\nonumber
\end{aligned}\end{equation}
where $Z_{i,j}$ is a contraction. Then, one can derive that
\begin{align}
&Y_i^{-1}X_{i,i}Y_i^{-1}+Y_j^{-1}X_{j,j}Y_j^{-1}\nonumber\\
&-Y_i^{-1}X_{i,j}Y_j^{-1}-Y_j^{-1}X_{j,i}Y_i^{-1}\nonumber\\
=&Y_i^{-1}\sqrt{X_{i,i}}\sqrt{X_{i,i}}Y_i^{-1}+Y_j^{-1}\sqrt{X_{j,j}}\sqrt{X_{j,j}}Y_j^{-1}\nonumber\\
&-Y_i^{-1}\sqrt{X_{i,i}}Z_{i,j}\sqrt{X_{j,j}}Y_j^{-1}-Y_j^{-1}\sqrt{X_{j,j}}Z_{i,j}^T\sqrt{X_{i,i}}Y_i^{-1}\nonumber\\
\geq&(Y_i^{-1}\sqrt{X_{i,i}}Z_{i,j}-Y_j^{-1}\sqrt{X_{j,j}})\nonumber\\
&\times (Y_i^{-1}\sqrt{X_{i,i}}Z_{i,j}-Y_j^{-1}\sqrt{X_{j,j}})^T\nonumber\\
\geq&0,
\end{align}
where the first inequality follows from the formula $XY^T+YX^T\leq XX^T+YY^T$ ($X$ and $Y$ have appropriate dimensions). With (85), the lemma can be proved by an argument similar to \cite{Niehsen1020907}. The details are omitted for brevity.
\end{proof}

\begin{proof}[The proof of (53)]
Based on the reformulation (45), one has
\begin{equation}\begin{aligned}
X(k+1)=\mathcal{F}_{k}(X(k)),
\end{aligned}\end{equation}
where
\begin{equation}\begin{aligned}
\mathcal{F}_{k}(X)\triangleq\mathcal{T}_{k}(X)+\Delta(k),
\end{aligned}\end{equation}
\begin{equation}\begin{aligned}
\mathcal{T}_{k}(X)\triangleq\sum^N_{i=1}&S_i(\sigma_k)^\mathrm{T}\big(\bar{A}(\sigma_k)^{-\mathrm{T}}-K(k)\big)X\\
&\times\big(\bar{A}(\sigma_k)^{-\mathrm{T}}-K(k)\big)^\mathrm{T}S_i(\sigma_k),
\end{aligned}\end{equation}
\begin{equation}\begin{aligned}
\Delta(k)\triangleq\sum^N_{i=1}&S_i(\sigma_k)^\mathrm{T}\big(K(k)\bar{Q}(\sigma_k)K^\mathrm{T}(k)+\bar{C}(\sigma_{k+1})^\mathrm{T}\\
&\times R(\sigma_{k+1})^{-1}\bar{C}(\sigma_{k+1})\big)S_i(\sigma_k),
\end{aligned}\end{equation}
\begin{equation}\begin{aligned}
K(k)\triangleq\bar{A}(\sigma_k)^{-\mathrm{T}}X(k)\big(X(k)+\bar{Q}(\sigma_k)\big)^{-1}.
\end{aligned}\end{equation}
Note that $\sigma_{iT+\vartheta}=\vartheta$ for $\vartheta\in\mathbb{T}$. With this fact, it follows from the optimality of the Kalman gain \cite{Sinopoli1333199} that for any $0\leq X(0)<\infty I$,
\begin{equation}\begin{aligned}
&\ \ \ X(iT+\vartheta;0,X(0))\\
&=\mathcal{F}_{iT+\vartheta-1}(X(iT+\vartheta-1;0,X(0)))\\
&\leq \mathcal{G}_{\vartheta-1}(X(iT+\vartheta-1;0,X(0))),\ \forall\vartheta\in\mathbb{T}\setminus\{0\},
\end{aligned}\end{equation}
\begin{equation}\begin{aligned}
&\ \ \ X(iT;0,X(0))\\
&=\mathcal{F}_{iT-1}(X(iT-1;0,X(0)))\\
&\leq \mathcal{G}_{T-1}(X(iT-1;0,X(0))).
\end{aligned}\end{equation}
Recursing (91) and (92) yields
\begin{equation}\begin{aligned}
X(iT+\vartheta;0,X(0))\leq\hat{\mathcal{G}}_\vartheta(X((i-1)T+\vartheta;0,X(0))).
\end{aligned}\end{equation}
Consequently, it follows from the fact that $\hat{\mathcal{S}}_\vartheta$ is a linear operator that
\begin{equation}\begin{aligned}
&X(iT+\vartheta;0,X(0))-X_\vartheta\\
\leq&\hat{\mathcal{G}}_\vartheta(X((i-1)T+\vartheta;0,X(0)))-\hat{\mathcal{G}}_\vartheta(X_\vartheta)\\
=&\hat{\mathcal{S}}_\vartheta(X((i-1)T+\vartheta;0,X(0)))-\hat{\mathcal{S}}_\vartheta(X_\vartheta)\\
=&\hat{\mathcal{S}}_\vartheta(X((i-1)T+\vartheta;0,X(0))-X_\vartheta).
\end{aligned}\end{equation}
Setting $X(0)=X_0+\Upsilon$ and recursing the inequality (94) yields (53). The proof is completed.
\end{proof}

\begin{proof}[The proof of $\Upsilon_\vartheta>0$] It follows from (49) that $\Upsilon_\vartheta>0$ for $\vartheta\in\mathbb{T}$ if $X_\vartheta>0$ for $\vartheta\in\mathbb{T}$. Thus, it remains to show that $X_\vartheta>0$ for $\vartheta\in\mathbb{T}$.

Let $\hat{P}(0|0)=bI$, one knows that $X(k;0,aI)^{-1}=\hat{P}(k|k)$ for $k\geq0$ if $0<a=\frac{1}{b}<\infty$. It follows from Theorem 1 that there exists $K>0$ such that $\hat{P}(k|k)\leq p_u I$ for $k\geq K$ and $0<b<\infty$, where $p_u=\max\{p_{u,1},p_{u,2},\cdots,p_{u,N}\}$. This implies that $X(K;0,aI)\geq p_u I$ for all $0<a<\infty$. Note that $X(K;0,aI)$ is depend continuously on $a$, thus one can obtain from a continuity argument \cite{horn2012matrix} that $X(K;0,0)\geq (p_u-\varepsilon) I$ for any $\varepsilon>0$. In this case, utilizing Lemma A.2 yields
\begin{equation}\begin{aligned}
X(k;0,0)>0,\ \forall k\geq K.
\end{aligned}\end{equation}
Consequently, $X_\vartheta>0$ for $\vartheta\in\mathbb{T}$.
\end{proof}

\begin{proof}[The proof of (79)]
Partition $X(k)$ conformally with $F(k)$ and denote $(i,j)$th sub-block of $X(k)$ as $X_{i,j}(k)\in\mathbb{R}^{n\times n}$.

Then, we will prove by an induction that $X_{i,i}(k)\leq P_{i}^{[\sigma_k]}$ for $i\in\mathbb{N}$ and $k\geq0$. At the initial moment, it is clear that $X_{i,i}(0)=0\leq P_{i}^{[\sigma_0]}$ for $i\in\mathbb{N}$. Suppose that $X_{i,i}(k-1)\leq P_{i}^{[\sigma_{k-1}]}$ for $i\in\mathbb{N}$. Then, for any $i\in\mathbb{N}$, one can derive from Lemma A.5 and the reformulation (45) that
\begin{equation}\begin{aligned}
&X_{i,i}(k)\\
=&\big(A(\sigma_{k-1})-K_{i}^{[\sigma_{k}]}C_i(\sigma_k)A(\sigma_{k-1})\big)\\
&\times\big(\sum_{j_1\in\mathbb{I}_i(\sigma_{k-1})}\sum_{j_2\in\mathbb{I}_i(\sigma_{k-1})}W_{i,j_1}^{[\sigma_{k-1}]}X_{j_1,j_2}(k-1)(W_{i,j_2}^{[\sigma_{k-1}]})^\mathrm{T}\big)\\
&\times\big(A(\sigma_{k-1})-K_{i}^{[\sigma_{k}]}C_i(\sigma_k)A(\sigma_{k-1})\big)^\mathrm{T}\\
&+\big(I-K_{i}^{[\sigma_k]}C_i(\sigma_k)\big)Q(\sigma_{k-1})\big(I-K_{i}^{[\sigma_k]}C_i(\sigma_k)\big)^\mathrm{T}\\
&+K_{i}^{[\sigma_k]}R_i(\sigma_k)(K_{i}^{[\sigma_k]})^\mathrm{T}\\
\leq&\big(A(\sigma_{k-1})-K_{i}^{[\sigma_{k}]}C_i(\sigma_k)A(\sigma_{k-1})\big)\hat{P}_i^{[\sigma_{k-1}]}\\
&\times\big(A(\sigma_{k-1})-K_{i}^{[\sigma_{k}]}C_i(\sigma_k)A(\sigma_{k-1})\big)^\mathrm{T}\\
&+\big(I-K_{i}^{[\sigma_k]}C_i(\sigma_k)\big)Q(\sigma_{k-1})\big(I-K_{i}^{[\sigma_k]}C_i(\sigma_k)\big)^\mathrm{T}\\
&+K_{i}^{[\sigma_k]}R_i(\sigma_k)(K_{i}^{[\sigma_k]})^\mathrm{T}\\
=&P_{i}^{[\sigma_{k}]},\nonumber
\end{aligned}\end{equation}
where we have used the equality $\bar{P}_i^{[\sigma_{k}]}=A(\sigma_{k-1})\hat{P}_i^{[\sigma_{k-1}]}A(\sigma_{k-1})^{\mathrm{T}}+Q(\sigma_{k-1})$ at the last equality. The proof is completed.
\end{proof}
\bibliographystyle{ieeetr}

\end{document}